\documentclass{article}

\usepackage{arxiv}

\usepackage[utf8]{inputenc} % allow utf-8 input
\usepackage[T1]{fontenc}    % use 8-bit T1 fonts
\usepackage{amsmath,amsfonts}
\usepackage{algorithmic}
\usepackage{algorithm}
\usepackage{array}
\usepackage[caption=false,font=normalsize,labelfont=sf,textfont=sf]{subfig}
\usepackage{textcomp}
\usepackage{stfloats}
\usepackage{url}
\usepackage{verbatim}
\usepackage{graphicx}
\usepackage{ntheorem}
\usepackage{soul}
\usepackage{epstopdf}
\soulregister\cite7
\soulregister\ref7
\usepackage{color}
\usepackage{booktabs}
\newcommand{\R}{\mathbb R}
\newcommand{\C}{\mathbb C}

\newtheorem{remark}{Remark}[section]

\newtheorem{lemma}{Lemma}[section]

\newtheorem{assumption}{Assumption}[section]

\newtheorem{proposition}{Proposition}[section]

\newtheorem{definition}{Definition}[section]

\newtheorem{theorem}{Theorem}[section]

\newtheorem*{proof}{Proof}[section]
\newcommand{\tabincell}[2]{\begin{tabular}{@{}#1@{}}#2\end{tabular}}

\allowdisplaybreaks[4]

\graphicspath{ {./images/} }

\title{A robust design of time-varying internal model principle-based control for ultra-precision tracking in a direct-drive servo stage}

\author{
 Yue Cao \\
  Department of Mechanical Engineering\\
  Tsinghua University\\
  Beijing 100084, China\\
   \And
  Zhen Zhang* \\
  Department of Mechanical Engineering\\
  Tsinghua University\\
  Beijing 100084, China\\
  \texttt{zzhang@tsinghua.edu.cn} \\
}

\begin{document}
\maketitle
\begin{abstract}
This paper proposes a robust design of the time-varying internal model principle-based control (TV-IMPC) for tracking sophisticated references generated by linear time-varying (LTV) autonomous systems. The existing TV-IMPC design usually requires a complete knowledge of the plant I/O (input/output) model, leading to the lack of structural robustness. To tackle this issue, we, in the paper, design a gray-box extended state observer (ESO) to estimate and compensate unknown model uncertainties and external disturbances. By means of the ESO feedback, the plant model is kept as nominal, and hence the structural robustness is achieved for the time-varying internal model. It is shown that the proposed design has bounded ESO estimation errors, which can be further adjusted by modifying the corresponding control gains. To stabilize the ESO-based TV-IMPC, a time-varying stabilizer is developed by employing Linear Matrix Inequalities (LMIs). Extensive simulation and experimental studies are conducted on a direct-drive servo stage to validate the proposed robust TV-IMPC with ultra-precision tracking performance ($\sim 60$nm {RMSE} out of $\pm80$mm stroke).
\end{abstract}

% keywords can be removed
%\keywords{First keyword \and Second keyword \and More}

\section{Introduction}
\label{intro}
{One of the central themes in the control of mechatronic systems is trajectory tracking, enabling many important applications to micro-/nano-precision manipulating systems, such as advanced machine tools~\cite{HUANG2022}, lithography machining process in semiconductors~\cite{haq2019}, and nano-scale inspection instrument~\cite{mitrovic2022}.}
%To improve the tracking performance and robustness of precision motion control systems, substantial control methodologies have been developed,
Continuing research efforts have been devoted to improving the tracking precision and robustness, for example, sliding mode control~\cite{ding2020,wang2022}, $H_\infty$ control~\cite{yan2019,zhang2021}, adaptive robust control~\cite{hu2020,chen2021}, iterative learning control~\cite{zheng2017,zhang2023}, etc. As one of the most investigated approaches, the \emph{Internal Model Principle} (IMP)~\cite{fran1976}, related to the \emph{output regulation} problem~\cite{isi1990,serrani2001}, is a powerful control methodology to asymptotically tracking/rejecting signals generated by  autonomous exogenous systems. Although the IMP-based control theory are {well established for} LTI and nonlinear systems~\cite{yang2019,HU2022556,mei2022}, the results for LTV systems remain open, mainly due to the fundamental challenges of constructing a robust time-varying internal model.

There have been advances of systematic design of the time-varying internal model principle-based control (TV-IMPC) for LTV systems~\cite{zhang2010,zhang2014,Paunonen2017}, which are applied to the control of camless engine valve~\cite{2014Time}. It is worth noting that the existing design still lacks structural robustness and cannot guarantee asymptotic tracking performance in the presence of plant model uncertainties and/or external disturbances. Currently, there is no general theoretical approach to solve this problem, and specific class of system models satisfying structural robustness is usually infeasible in practice. This motivates us to investigate an alternative to achieve the robustness of TV-IMPC.

One straightforward idea to this problem is resorting an adaptive control-based design to estimate the actual plant dynamics in real-time~\cite{schmidt2020,LIU202281}. This method, however, significantly affect the transient behavior of the TV-IMPC controller, leading to time-varying system stabilization task more challenging. An alternative approach is to adopt observer-based design to estimate the model uncertainties and external disturbances, and keep the plant model as nominal through  feedback. This method offers a simple control structure yet improved robustness of the time-varying internal model.

In this thread, the \emph{Disturbance Observer-Based Control} (DOBC) methods are frequently utilized for external disturbance estimation~\cite{chen2015,rouhani2021}. This method, however, has a limitation as it requires a precise  plant model, and thus it cannot directly estimate the model uncertainties. To overcome this limitation, we consider the \emph{Extended State Observer} (ESO)~\cite{li2019,zhao2021,WU2022}, which is able to estimate both model uncertainties and disturbances simultaneously.

The ESO is commonly used as a component of the \emph{Active Disturbance Rejection Control} (ADRC)~\cite{gao2006,han2009}. In the ESO structure, both the unmodeled system dynamics and the external disturbances are treated as an extra system state, and then a Luenberger observer is designed for the augmented system, which enables the estimation of lumped model uncertainties and disturbances. As a result, ESO can be used in plant model compensation to improve system robustness~\cite{sayem2016,wang2017}.

The most widely-used ESO is designed in a black-box manner, which is model-free, and might result in estimation degradation under non-negligible model uncertainties or disturbances. Alternatively, if a sufficiently accurate plant model is available, it can be incorporated into the ESO design to enhance the estimation performance. This approach is known as the gray-box ESO~\cite{zhang2016}. Based on the above discussion, we, in this paper, utilize the gray-box ESO to achieve an enhanced estimation of plant model uncertainties and external disturbances.

The contributions of this study are listed as follows:
\begin{itemize}
\item A robust design for the TV-IMPC methodology is proposed. The structural robustness of the time-varying internal model is achieved by constructing a nominal plant model via the ESO feedback.
\item The ESO is designed in the gray-box fashion to better estimate the potential model uncertainties and external disturbances, bringing improved estimation accuracy and overall closed-loop system stability.
\item The experimental results validate the theoretical analysis and demonstrate an ultra-precision tracking performance, of which the RMSE is $61.14$nm within $\pm80$mm motion stroke.
\end{itemize}

The rest of the paper is structured as follows. In Section~\ref{pro}, we outline the control problem for time-varying systems with unknown model uncertainties and external disturbances. Section~\ref{design} presents the design and analysis of the robust TV-IMPC, including the time-varying internal model controller, the gray-box ESO, and the time-varying stabilizer. Section~\ref{exp} provides extensive simulation and experimental results that demonstrate the outstanding tracking performance of the proposed robust design. The paper ends up with the conclusion in Section~\ref{conclusion}.
\section{Problem formulation}\label{pro}
{Consider a discrete LTI plant system of the form}
\begin{equation}\label{plant}
\begin{array}{rcl}
x(k+1)&=&Ax(k)+Bu(k)+Ed(k)\,\\[1mm]
y(k)&=&Cx(k)\,\\[1mm]
e(k)&=&y(k)-r(k)\,,
\end{array}
\end{equation}
with state $x\in\R^n$, control input $u\in\R$, external disturbances and model uncertainties $d\in\R^n$, output $y\in\R$, and tracking error $e\in \R$.

The reference $r\in \R$ to be tracked is generated by an LTV exosystem of the form
\begin{equation}\label{exotime}
\begin{array}{rcl}
 w(k+1)&=&S(k)w(k)\\[1mm]
r(k)&=&Q(k)w(k)\,,
\end{array}
\end{equation}
with exogenous state $w\in\R^{\rho}$.

The following assumptions are made for the systems~(\ref{plant}) and~(\ref{exotime}).
\begin{assumption}
\label{ctrl}
%The plant model is controllable, observable and stable (if not, one can stabilize it first).
The pair $(A,\, B)$ is controllable, and the pair $(A,\, C)$ is observable, and $x(k+1)=Ax(k)$ is asymptotically stable (if not, one can stabilize it).
\end{assumption}
\begin{assumption}
\label{lump}
The external disturbances and the model uncertainties $d(k)$ is bounded and can be lumped in the input channel as $d_l(k)\in\R$.
\end{assumption}

With Assumption~\ref{ctrl} and~\ref{lump}, system~(\ref{plant}) can be written in the following controllable canonical form as
\begin{equation}\label{dplant}
\begin{array}{rcl}
x(k+1)&=&A_{\rm c}x(k)+B_{\rm c}u(k)+E_{\rm c}d_l(k)\,\\[1mm]
y(k)&=&C_{\rm c}x(k)\,\\[1mm]
e(k)&=&y(k)-r(k)\,,
\end{array}
\end{equation}
where
\begin{equation*}\label{para}
\begin{array}{rcl}
A_{\rm c}&=&\begin{bmatrix} 0 & 1 & \cdots& 0\\
\vdots & \vdots & \ddots & \vdots\\
0 & 0& \cdots & 1\\
a_0 & a_1 & \cdots & a_{n-1}\end{bmatrix}\,,\\[10mm]
B_{\rm c}&=&\begin{bmatrix} 0 & \cdots & 0 &1\end{bmatrix}^{\top}\,,\\[1mm]
E_{\rm c}&=&\begin{bmatrix} 0 & \cdots & 0 &1/b\end{bmatrix}^{\top}\,,
\end{array}
\end{equation*}
and control input gain $b>0$.

\begin{definition}
\label{pbm}
The tracking control problem is to find a compensator fed only by the measurement of $y(k)$, such that there exists a certain function $\Delta(k)\in\R$ satisfying $\lim\nolimits_{k\rightarrow\infty}|\Delta(k)|\le\alpha$, where $\alpha$ is a positive constant, and the tracking error $e(k)$ satisfies $\lim\nolimits_{k\rightarrow\infty}(e(k)-\Delta(k))=0$.
\end{definition}

In the following section, we will eliminate the term $E_{\rm c}d_l(k)$ in the dynamics~(\ref{dplant}), so that the dynamics~(\ref{dplant}) can be transformed into a nominal one, which achieves a robust design of time-varying internal model.

\section{{Robust TV-IMPC}}\label{design}
\begin{figure*}[ht]
  \centering
  \includegraphics[width=\hsize]{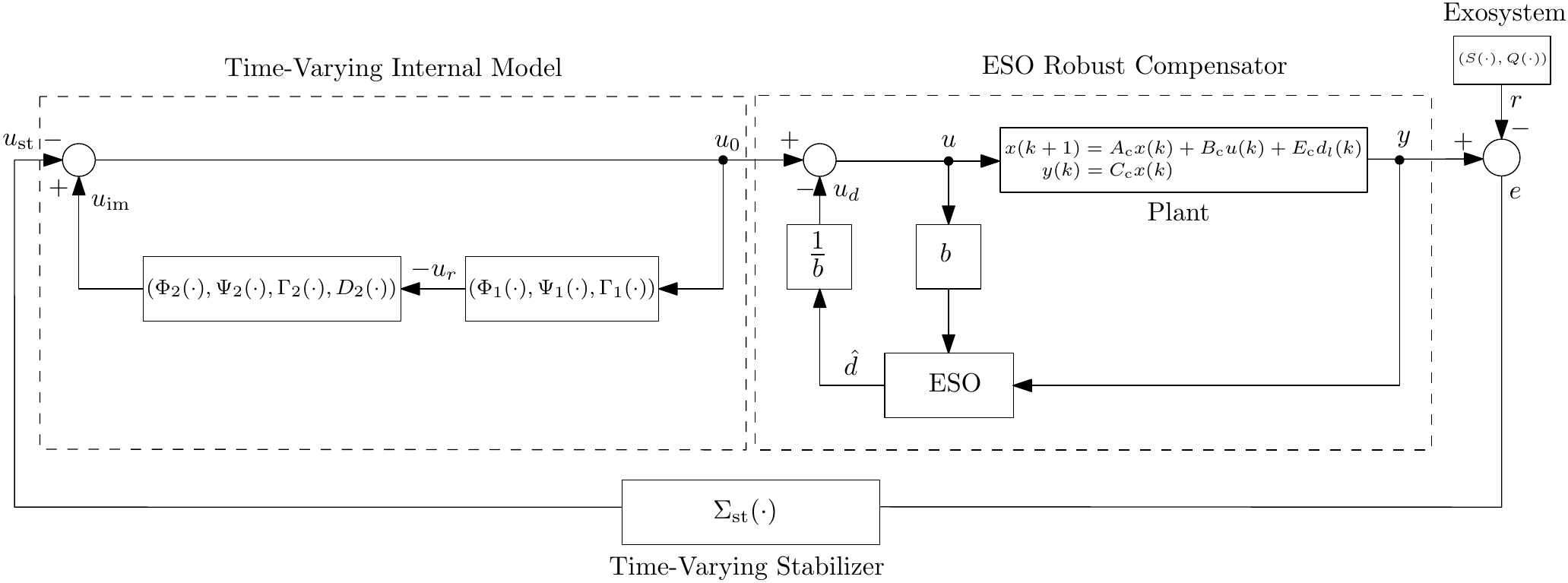}\\
  \caption{The block diagram of the robust TV-IMPC.}\label{dia}
\end{figure*}
The block diagram of the proposed robust TV-IMPC system is illustrated in Figure~\ref{dia}. And the detailed design includes a time-varying internal model, a gray-box ESO and a time-varying stabilizer. And {the motivation of  such a design is as follows.

Recall that the existing TV-IMPC design requires a complete knowledge of the plant I/O (input/output) model~\cite{zhang2010}. As a result, the structural robustness of the TV-IMCC is not ensured if there is a deviation between the actual plant model and the nominal one used in the controller design. To tackle this problem, the ESO, in this paper, is utilized to estimate and compensate the model uncertainties and disturbances in real-time. Therefore, the plant model is kept as nominal by means of the ESO feedback, and hence the structural robustness is achieved for the time-varying internal model.}

To proceed, we begin with the design of {nominal} time-varying internal model.
\subsection{Preliminaries of time-varying internal model design}\label{imc}
To solve the tracking control problem in Definition~\ref{pbm} by the internal model principle-based approach, a desired input
\begin{equation}\label{uff}
u_{\rm ff}(k)=R(k)w(k)
\end{equation}
needs to be found to make the tracking error converge~\cite{zhang2010}. Note that exosystem state $w$ is not available for feedback. %Combining equation~(\ref{uff}) with the exosystem dynamics~(\ref{exotime}), the following system
In order to construct $u_{\rm ff}$, a instrumental step is to {\em immerse} the exosystem~(\ref{exotime}) with output~(\ref{uff}),
\begin{equation}\label{exotime2}
\begin{array}{rcl}
 w(k+1)&=&S(k)w(k)\\[1mm]
u_{\rm ff}(k)&=&R(k)w(k)
\end{array}
\end{equation}
into the following $n_\xi$-order system
\begin{equation}\label{ims}
\begin{array}{rcl}
\xi(k+1)&=&\Phi(k)\xi(k)\\[1mm]
u_{\rm im}(k)&=&\Gamma(k)\xi(k)\,,
\end{array}
\end{equation}
whose output includes every output of system~(\ref{exotime2}). The definition of {\em system immersion} is as follows.
\begin{definition}
The system~(\ref{exotime2}) is said to be immersed into system~(\ref{ims}), if there exists a smooth mapping $U:\R\mapsto\R^{n_\xi\times\rho}$, satisfying
\begin{equation}\label{map}
\begin{array}{rcl}
U(k+1)+U(k)S(k)&=&\Phi(k)U(k)\\[1mm]
R(k)&=&\Gamma(k)U(k)\,.
\end{array}
\end{equation}
\end{definition}

To find a system immersion, system~(\ref{ims}), in our design,  is constructed by two time-varying controllers, i.e.
\begin{equation}\label{ims1}
\begin{array}{rcl}
\xi_1(k+1)&=&\Phi_1(k)\xi_1(k)+\Psi_1(k)u_0(k)\\[1mm]
u_r(k)&=&\Gamma_1(k)\xi_1(k)\,,
\end{array}
\end{equation}
and
\begin{equation}\label{ims2}
\begin{array}{rcl}
\xi_2(k+1)&=&\Phi_2(k)\xi_2(k)+\Psi_2(k)(-u_r(k))\\[1mm]
u_{\rm im}(k)&=&\Gamma_2(k)\xi_2(k)+D_2(k)(-u_r(k))\,,
\end{array}
\end{equation}
with the state $(\xi_1,\xi_2)\in(\R^{\rho},\R^{\rho-1})$, embedded input $u_r\in\R$, and controller output $u_{\rm im}\in\R$.
\begin{lemma}\label{l2}
The time-varying controller (\ref{ims1})-(\ref{ims2}) can be obtained by solving the following time-varying Sylvester equation
\begin{equation}\label{se}
\begin{bmatrix}\mathcal{O}_{\Phi_1(k)} & \mathcal{C}_{\Psi_1(k)} \end{bmatrix} \begin{bmatrix}1\\q(k)\\p(k)\end{bmatrix}= \mathcal{O}_{S(k)}\begin{bmatrix}1\\q(k)\end{bmatrix}\,.
\end{equation}
\end{lemma}
\begin{proof}
The proof is referred to the results in Ref.~\cite{zhang2010}. \hfill $\square$
\end{proof}
Specifically, $q(k)\in\R^{\rho-1}$ and $p(k)\in\R^{\rho}$ are the parameter vectors in the controller canonical form of $\Phi_2(\cdot)$ and $\Gamma_2(\cdot)$, {respectively}. The operator $\mathcal{O}_{S(k)}\in\R^{(2\rho-1)\times\rho}$ is defined as
\begin{equation}\label{os}
\mathcal{O}_{S(k)}=
\begin{bmatrix}
\alpha(k) & 1& 0 & \cdots & 0 \\
 0 & \alpha(k) & 1 & \ddots & \vdots \\
  \vdots & \ddots & \ddots & \ddots & 0 \\
   0 & \cdots & 0 & \alpha(k) & 1 \\
    0 & \cdots \ & \cdots & 0 &\alpha(k)
    \end{bmatrix}\,,
\end{equation}
with $\alpha(k)\in\R^{\rho}$ collecting the coefficients of the exosystem matrix $S(k)$ in its observer canonical form $S_{\rm o}(k)$, and $\mathcal{O}_{\Phi_1(k)}\in\R^{(2\rho-1)\times\rho}$ is defined similarly with $\mathcal{O}_{S(k)}$; and the operator $\mathcal{C}_{\Psi_1(k)}\in\R^{(2\rho-1)\times\rho}$ is defined as
\begin{equation}\label{cp}
\mathcal{C}_{\Psi_1(k)}=\begin{bmatrix} \psi_{\rho-1} & \cdots & \begin{bmatrix} 0_{\rho-2} \\ \psi_1 \end{bmatrix} & \begin{bmatrix} 0_{\rho-1} \\ \psi_0 \end{bmatrix} \end{bmatrix} \,,
\end{equation}
with
\begin{equation*}
\psi_{i+1}= \begin{bmatrix} \psi_i \\ 0 \end{bmatrix},\quad \psi_0=\Psi_1(k),\quad i=0,1,\cdots,\rho-2\,.
\end{equation*}

\begin{proposition}\label{io}
For the time-varying internal model triplet \[(\Phi_{1}(\cdot),~\Psi_{1}(\cdot),~\Gamma_{1}(\cdot))\,,\]it should have the same I/O  as that of the plant model~(\ref{dplant}) (from $u$ to $y$).
\end{proposition}

\begin{proposition}\label{immersion}
With the internal model design~(\ref{ims1})-(\ref{ims2}), the Sylverster equation~(\ref{se}) and  Proposition~\ref{io}, the system~(\ref{exotime2}) is immersed into system~(\ref{ims})
\begin{equation*}\label{imsi}
\begin{array}{rcl}
\xi(k+1)&=&\Phi(k)\xi(k)\,\\[1mm]
u_{\rm im}(k)&=&\Gamma(k)\xi(k)\,,\\[1mm]
\end{array}
\end{equation*}
where
\begin{equation*}
\begin{array}{rcl}
\Phi(k)&=&\begin{bmatrix}\Phi_1(k)-\Psi_1(k)D_2(k)\Gamma_1(k) & \Psi_1(k)\Gamma_2(k)\\-\Psi_2(k)\Gamma_1(k)&\Phi_2(k)\end{bmatrix}\,,\\[5mm]
\Gamma(k)&=&\begin{bmatrix}-D_2(k)\Gamma_1(k)&\Gamma_2(k)\end{bmatrix}\,.\\[1mm]
\end{array}
\end{equation*}
And there exists a smooth mapping $\Sigma:\R\mapsto\R^{2\rho-1\times\rho}$, satisfying
\begin{equation}\label{map2}
\begin{array}{rcl}
\Sigma(k+1)+\Sigma(k)S(k)&=&\Phi(k)\Sigma(k)\\[1mm]
R(k)&=&\Gamma(k)\Sigma(k)\,.
\end{array}
\end{equation}
\end{proposition}

\begin{proof}
{By noting the discrete-time version of phase-variable canonical forms in Ref.~\cite{1969Ra},} the proof is then followed by the results in Proposition 3.3 and Lemma 4.1 of Ref.~\cite{zhang2010}. \hfill $\square$
\end{proof}

As long as the immersion is found by Proposition~\ref{immersion}, and the augmented system of the controller~(\ref{ims}) and the plant model~(\ref{dplant}) is stabilized, the tracking error is converged thanks to the time-varying internal model.
\subsection{Design of gray-box ESO}\label{esod}
It is worth noting that the design of the triplet $(\Phi_{1}(\cdot),~\Psi_{1}(\cdot),~\Gamma_{1}(\cdot))$ in Proposition~\ref{io} is not straightforward, because that the system~(\ref{dplant}) contains unknown variable $d_l(k)$. If the nominal form of system~(\ref{dplant}) is utilized to construct $(\Phi_{1}(\cdot),~\Psi_{1}(\cdot),~\Gamma_{1}(\cdot))$ as
\begin{equation}\label{ims1d}
(\Phi_{1}(\cdot),~\Psi_{1}(\cdot),~\Gamma_{1}(\cdot))=(A_{\rm c},~B_{\rm c},~C_{\rm c})\,,
\end{equation}
it remains a feedback control to compensate the model uncertainties and external disturbances $d_l(k)$ in the plant model~(\ref{dplant}). To this end, we design an ESO-based compensator.

Note that the design of ESO can be classified into different ones according to the usage of the system information. Specifically, the design of black-box ESO does not consider the information of plant dynamics in triplet $(A_{\rm c}, B_{\rm c}, C_{\rm c})$. To improve the estimation accuracy, the gray-box ESO is chosen, in this paper, by using the given plant model information.

Based on the system dynamics~(\ref{dplant}), the gray-box ESO is designed as follows
\begin{equation}\label{obs}
\begin{split}
\hat x(k+1)&=A_{\rm c}\hat x(k)+B_{\rm c}u(k)+E_{\rm c}\hat d_l(k)-L_1(\hat y(k)-y(k))\\[2mm]
\hat d(k+1)&=\hat d_l(k)-L_2(\hat y(k)-y(k))\\[2mm]
\hat y(k)&=C_{\rm c}\hat x(k)\,,
\end{split}
\end{equation}
where $L_1\in\R^n$ and $L_2\in\R$ are observer gains. Specifically, the feedback of the ESO is chosen as $u_d(k)=\dfrac{\hat d_l(k)}b$. Then the plant model dynamics with feedback becomes
\begin{equation}\label{sf}
\begin{split}
x(k+1)&= A_{\rm c}x(k)+B_{\rm c}u_0(k)+E_{\rm c}(d_l(k)-\hat d_l(k))\\[2mm]
y(k)&=C_{\rm c}x(k)\,,
\end{split}
\end{equation}
by noting that $u(k)=u_0(k)-u_d(k)$. System~(\ref{sf}) with the ESO feedback should behave similar to the following nominal system
\begin{equation}\label{ns}
\begin{split}
x_{\rm n}(k+1)&=A_{\rm c}x_{\rm n}(k)+B_{\rm c}u_0(k)\\[2mm]
y_{\rm n}(k)&=C_{\rm c}x_{\rm n}(k)\,.
\end{split}
\end{equation}

\begin{remark}
\label{bound}
To clarify the similarity between system~(\ref{sf}) and~(\ref{ns}), it is necessary to verify the boundedness of their output difference $y(k)-y_{\rm n}(k)$.
\end{remark}
\subsection{Analysis of estimation error }\label{est}
Define the ESO estimation errors as $\tilde x(k)=\hat x(k)-x(k)$, and $~\tilde d_l(k)=\hat d_l(k)-d_l(k)$. According to~(\ref{dplant}) and~(\ref{obs}), it is found that
\begin{equation}\label{error}
\begin{split}
\tilde x(k+1)&=(A_{\rm c}-L_1C_{\rm c})\tilde x(k)+E_{\rm c}\tilde d_l(k)\\[2mm]
\tilde d_l(k+1)&=-L_2C_{\rm c}\tilde x(k)+\tilde d_l(k)+(d_l(k+1)-d_l(k))\,.
\end{split}
\end{equation}
%Thus
%\begin{equation}\label{erd}
%\begin{split}
%\begin{bmatrix}\tilde x(k+1)\\\tilde d_l(k+1)\end{bmatrix}&=\begin{bmatrix}G_{\rm c}-L_1C_{\rm c} & E_{\rm c}\\-L_2C_{\rm c} &I\end{bmatrix}\begin{bmatrix}\tilde x(k)\\ \tilde d_l(k)\end{bmatrix}+\begin{bmatrix}0\\d_l(k+1)-d_l(k)\end{bmatrix}.
%\end{split}
%\end{equation}
Error dynamics~(\ref{error}) can be put in a more compact form as
\begin{equation}\label{erd3}
\zeta(k+1)=A_{\rm a}\zeta(k)+B_{\rm a}d_{ld}(k)\,,
\end{equation}
with $\zeta(k):=\begin{bmatrix}\tilde x(k+1)\\\tilde d_l(k+1)\end{bmatrix}, A_{\rm a}:=\begin{bmatrix}A_{\rm c}-L_1C_{\rm c} & E_{\rm c}\\-L_2C_{\rm c} &I\end{bmatrix}, d_{ld}(k):=d_l(k+1)-d_l(k)$, and $B_{\rm a}=\begin{bmatrix}0 & \cdots & 0 & 1\end{bmatrix}^{\top}$.
According to Assumption~\ref{lump}, it is obvious that $|d_{ld}(k)|\le\delta$, where $\delta\in\R$ is a positive constant. Equation~(\ref{erd3}) yields that
\begin{equation}\label{ert}
\zeta(k)=A_{\rm a}^k\zeta(0)+\sum_{j=0}^{k-1}A_{\rm a}^{k-j-1}B_{\rm a}d_{ld}(j)\,.
\end{equation}
Note that in the design of ESO, it is necessary to make system~(\ref{erd3}) stable by calculating the observer gains $L_1$ and $L_2$ to make $A_{\rm a}$ Hurwitz. Thus, the term $A_{\rm a}^k\zeta(0)$ in equation~(\ref{ert}) converges to zero as $k\rightarrow \infty$. Therefore, $\zeta(k)$ converges to
\begin{equation}\label{ert1}
\lim_{k\rightarrow \infty}\zeta(k)=\sum_{j=0}^{k-1}A_{\rm a}^{k-j-1}B_{\rm a}d_{ld}(j)\,.
\end{equation}
Since $A_{\rm a}$ is Hurwitz, the condition $\rho(A_{\rm a})<1$ is satisfied, where $\rho(\cdot)$ represents the spectral radius of a matrix. With this, the following inequality can be obtained as
\begin{equation}\label{ert2}
\begin{split}
\lim_{k\rightarrow \infty}|\zeta(k)| \le& \lim_{k\rightarrow \infty}\delta|\sum_{j=0}^{k-1}A_{\rm a}^{k-j-1}B_{\rm a}|\\[2mm]
\le&\delta|(I-A_{\rm a})^{-1}B_{\rm a}|\\[2mm]
=&\delta L_2^{-1}|(E_cC_c)^{-1}B_{\rm a}|\,,\\[2mm]
\end{split}
\end{equation}
by noting that $|d_{ld}(k)|\le\delta$. Therefore, the estimation error of ESO~(\ref{obs}) is bounded by regulating the ESO gains $L_1$ and $L_2$ to stabilize the error dynamics~(\ref{erd3}), and the estimation error can be further reduced by increasing the ESO gain $L_2$.

Next, denote
\begin{equation}\label{deltad1}
\varepsilon(k):=x_{\rm n}(k)-x(k),
\end{equation}
with the state of the nominal system~(\ref{ns}) $x_{\rm n}(k)$. Combining Equation~(\ref{sf}) and (\ref{ns}) yields
\begin{equation}\label{nes}
\varepsilon(k+1)=A_{\rm c}\varepsilon(k)+E_{\rm c}\tilde d_l(k)\,.
\end{equation}
Note that $\tilde d_l(k)$ is bounded according to~(\ref{ert2}), thus
\begin{equation}\label{nest}
\varepsilon(k)=A_{\rm c}^k\varepsilon(0)+\sum_{j=0}^{k-1}A_{\rm c}^{k-j-1}E_{\rm c}\tilde d_l(j)\,.
\end{equation}
Similar to the induction in~(\ref{ert2}), it is obtained that
\begin{equation}\label{ert3}
\lim_{k\rightarrow \infty}|\varepsilon(k)|\le\sup(|\tilde d_l(k)|)|(I-A_c)^{-1}E_c|\,.
\end{equation}
Hence, $\varepsilon(k)$ is also bounded, and $|\varepsilon(k)|$ can be made arbitrarily small by pole placement of the nominal system dynamics $(A_c, B_c, C_c)$. Therefore, it is easy to know that
\begin{equation}\label{nest2}
\Delta(k):=-C_c\varepsilon(k)=y(k)-y_{\rm n}(k)
\end{equation}
is bounded, i.e.
\begin{equation}\label{ert4}
\lim_{k\rightarrow \infty}|\Delta(k)|\le\sup(|\tilde d_l(k)|)|(I-A_c)^{-1}E_cC_c|\,,
\end{equation}
and $\Delta(k)$ can be recognized as the additive uncertainty. Note that $\Delta(k)$ satisfies the constraints in Definition~\ref{pbm}, and verifies the similarity between system~(\ref{sf}) and~(\ref{ns}) in Remark~\ref{bound}.

With~(\ref{nest2}), system~(\ref{dplant}) can be transformed to a nominal one with the additive uncertainty $\Delta(k)$ as follows
\begin{equation}\label{dnplant}
\begin{array}{rcl}
x_{\rm n}(k+1)&=&A_{\rm c}x_{\rm n}(k)+B_{\rm c}u_0(k)\,\\[1mm]
y_{\rm n}(k)&=&C_{\rm c}x_{\rm n}(k)\,\\[1mm]
e(k)-\Delta(k)&=&y_{\rm n}(k)-r(k)\,.
\end{array}
\end{equation}

According to the above analysis, the nominal design~(\ref{ims1d}) can be utilized in the controller~(\ref{ims1}) for the system~(\ref{dnplant}). And we are now in position to design a stabilizer for system~(\ref{dnplant}), time-varying internal model controller~(\ref{ims1})-(\ref{ims2}), and ESO~(\ref{obs}).
\subsection{Design of the stabilizer}\label{stabd}
The stabilize is designed for the augmented system of the time-varying internal model controller~(\ref{ims1})-(\ref{ims2}) and the nominal plant model~(\ref{dnplant}) as follows,
\begin{align}\label{augment}
\begin{split}
\begin{bmatrix}\xi_2(k+1)\\x_{\rm n}(k+1) \end{bmatrix} =&\begin{bmatrix}\Phi_2(k) & -\Psi_2(k)C_{\rm c}\\B_{\rm c}\Gamma_2(k)&A_{\rm c}- B_{\rm c}D_2(k)C_{\rm c} \end{bmatrix} \begin{bmatrix}\xi_2(k)\\x_{\rm n}(k) \end{bmatrix} +\begin{bmatrix}0\\B_{\rm c} \end{bmatrix}u_{\rm st}(k)\,.%\\[3mm]
%&+ \begin{bmatrix}0\\H_{\rm c} \end{bmatrix}u_{\rm st}(k)\,.
\end{split}
\end{align}
System~(\ref{augment}) can be rewritten as~(\ref{augsys}) by providing the following lemma.
\begin{lemma}\label{stabilization}
To stabilize the augmented time-varying system~(\ref{augment}), it is sufficient to stabilize the following one
\begin{equation}\label{augsys}
\begin{array}{rcl}
x_{\rm o}(k+1)&=&F(k)x_{\rm o}(k)+Gu_{\rm st}(k)\,\\[1mm]
y_{\rm n}(k)&=&C_{\rm o}x_{\rm o}(k)\,,
\end{array}
\end{equation}
where
\begin{equation*}
F(k)=\begin{bmatrix} -\alpha(k)_{\rho\times1} & I_{(n-1)\times(n-1)}\\0_{(n-\rho)\times1}&0_{1\times(n-1)} \end{bmatrix}\,,\mbox{and }\,G=B_c\,,
\end{equation*}
with $\alpha(k)$ collecting the coefficients of the first column of exosystem in its observer canonical form, and $C_{\rm o}=\begin{bmatrix}1&0&\cdots&0\end{bmatrix}$.
\end{lemma}
\begin{proof}
The proof is referred to Lemma 3.1 of Ref.~\cite{zhang2014}.
\hfill $\square$
\end{proof}

Note that system~(\ref{augsys}) can be split as
\begin{equation}\label{splitsys}
\begin{array}{rcl}
x_{\rm o1}(k+1)&=&F_{11}(k)x_{\rm o1}(k)+F_{12}(k)x_{\rm o1}(k)+G_1u_{\rm st}(k)\,\\[1mm]
x_{\rm b}(k+1)&=&F_{21}(k)x_{\rm b}(k)+F_{22}(k)x_{\rm b}(k)+G_2u_{\rm st}(k)\,\\[1mm]
y_2(k)&=&x_{\rm o1}(k)\,,
\end{array}
\end{equation}
where $x_{\rm o1}\in\R$ is the first state of $x_{\rm o}$, and $x_{\rm b}\in\R^{n-1}$ collects the rest $(n-1)$ states of $x_{\rm o}$. %, and $A_{11}(k),A_{21}(k),A_{12},A_{22},B_1(k),B_2(k)$ are the corresponding splits of $A(k)$ and $B(k)$.
We introduce a reduced order observer of the state $x_{\rm b}$ as follows:
\begin{equation}\label{ob}
\begin{array}{rcl}
\hat z(k+1)&=&(F_{22}-HF_{12})\hat z(k)+(G_2-HG_1)u_{\rm st}(k)\\[1mm]
&&+((F_{22}-HF_{12})H+F_{21}(k)-HF_{11}(k))y_2(k)\\[1mm]
\hat x_{\rm b}(k)&=&\hat z(k)+Hy_2(k)\,,
\end{array}
\end{equation}
where $H$ is the output injection gain of the observer. With the estimation $\hat x_{\rm b}(k)$, one can stabilize system~(\ref{splitsys}) via the following stabilizer
\begin{equation*}
u_{\rm st}=\begin{bmatrix}K_1(k) & K_2(k)\end{bmatrix}\begin{bmatrix}x_{\rm o1}(k)\\ \hat x_{\rm b}(k)\end{bmatrix}\,.
\end{equation*}
Hence, the remaining task is to make the close-loop system stable
\begin{equation}\label{stabsys}
x_{\rm o}(k+1)=(F(k)+GK(k))x_{\rm o}(k),
\end{equation}
with a stabilizer gain $K(k)=\begin{bmatrix}K_1(k) & K_2(k)\end{bmatrix}\in \begin{bmatrix}\R & \R^{1\times(n-1)}\end{bmatrix}$. To solve $K(k)$, we assume that $F(k)$ has the following property.
\begin{assumption}
\label{polyt}
The matrix $F(k)$ belongs to a polytype, that is
\begin{equation}\label{poly}
F(k)=F(\sigma(k))=\sum\limits_{i=1}^N\sigma_i(k)F_i\,,
\end{equation}
where
\begin{equation*}
\sigma_i(k)\ge0,~~\sum\limits_{i=1}^N\sigma_i(k)=1\,,
\end{equation*}
and $F_i$'s are constant matrices.
\end{assumption}

The following Lemma is provided to calculate the parameters of $K(k)$.
\begin{lemma}\label{polytype}
If there exist symmetric matrices $Q_i>0,~Q_j>0$, matrices $M_i,~T_i$, constant $\gamma>1$, and the following matrix inequalities are solvable
\begin{equation}\label{lmi}
\begin{bmatrix} -M_i-M_i^{\top} +Q_i & * & * & *\\-(F_iM_i+BT_i) & -Q_j & * & *\\C_{\rm o}&0&-I& *\\C_{\rm o}&0&0&-(\gamma^2-1)I\end{bmatrix}<0\,,
\end{equation}
for all $i=1,2,\cdots,N$ and $j=1,2,\cdots,N$, then the feedback gain to stabilize the augment system~(\ref{stabsys}) is $K(k)=\sum_{i=1}^N\sigma_i(k)T_iM_i^{-1}$.
\end{lemma}
\begin{proof}
Consider a Lyapunov function candidate
\begin{equation}\label{lya}
V(x_{\rm o}(k),\sigma(k))=x^{\top}_{\rm o}P(\sigma(k))x_{\rm o}(k),
\end{equation}
where $P(k)$ is a symmetric positive-definite matrix for the close-loop system~(\ref{stabsys}), and denote $A_{\rm st}(k):=F(k)+GK(k)$. To make the presentation concise, we drop index $k$, while keeping $k+1$ in the following notations. Since the reference plays no role in stabilization, the reference $r(k)$ is set as zero, and $e=C_{\rm o}x_{\rm o}+\Delta$ is yielded according to~(\ref{dnplant}) and~(\ref{augsys}). Then we have
\begin{equation}\label{dv}
\begin{split}
&V(x_{\rm o}(k+1),\sigma(k+1))-V(x_{\rm o},\sigma)+e^{\top}e-\gamma^2\Delta^{\top}\Delta\\[2mm]
&=\begin{bmatrix} x_{\rm o}\\\Delta\end{bmatrix}^{\top}\begin{bmatrix} X & C_{\rm o}^{\top}\\C_{\rm o} & I-\gamma^2I\end{bmatrix}\begin{bmatrix} x_{\rm o}\\\Delta\end{bmatrix}\,,
\end{split}
\end{equation}
with $X:=-P+A_{\rm st}^{\top}P(k+1)A_{\rm st}+C_{\rm o}^{\top}C_{\rm o}$. To make~(\ref{dv}) negative definite, it is equivalent to say that
\begin{equation}\label{dv2}
-P+A_{\rm st}^{\top}P(k+1)A_{\rm st}+C_{\rm o}^{\top}C_{\rm o}+C_{\rm o}^{\top}(\gamma^2I-I)^{-1}C_{\rm o}<0.
\end{equation}
Meanwhile, note that
\begin{equation}\label{dv3}
\begin{split}
-P+A_{\rm st}^{\top}P(k+1)A_{\rm st}%\\[2mm]
=-P+A_{\rm st}^{\top}P^{\top}(k+1)P^{-1}(k+1)P(k+1)A_{\rm st}\,,
\end{split}
\end{equation}
hence one can take the Schur complement on~(\ref{dv2}) and rewrite the LMI as
\begin{equation}\label{lmi2}
\begin{bmatrix} -P& * & * & *\\-P(k+1)A_{\rm st} & -P(k+1) & * & *\\C_{\rm o}&0&-I& *\\C_{\rm o}&0&0&-(\gamma^2-1)I\end{bmatrix}<0\,.
\end{equation}
Noting that $Q=P^{-1}$, the above matrix inequalities read as
\begin{equation}\label{lmi3}
\begin{bmatrix} -Q^{-1}& * & * & *\\-Q^{-1}(k+1)A_{\rm st} & -Q^{-1}(k+1) & * & *\\C_{\rm o}&0&-I& *\\C_{\rm o}&0&0&-(\gamma^2-1)I\end{bmatrix}<0\,.
\end{equation}
which implies inequality~(\ref{lmi}) by using the results in Ref.~\cite{daafouz2001}. Recall~(\ref{dv}), we obtain that
\begin{equation}\label{dv4}
V(x_{\rm o}(k+1),\sigma(k+1))-V(x_{\rm o},\sigma)+e^{\top}e-\gamma^2\Delta^{\top}\Delta<0\,.
\end{equation}
Summating on both sides of the above inequality, it is seen that $|e|<\gamma|\Delta|$, which completes the proof. \hfill $\square$
\end{proof}

With the internal model controller in Section~\ref{imc}, the ESO in Section~\ref{esod}, and the stabilizer in Section~\ref{stabd}, we are in position to present the robust TV-IMPC.

\subsection{Design of robust TV-IMPC}
The robust TV-IMPC is stated in the following theorem.
\begin{theorem}
If Assumption~\ref{polyt} holds, then the tracking error $e(k)$ of system~(\ref{exotime})-(\ref{dplant}) satisfies
\[
\lim_{k\to\infty}(e(k)-\Delta(k))=0\,,
\]
with $\Delta(k)$ being defined in~(\ref{nest2}), by applying the following TV-IMPC controller,
\begin{align*}\label{stab}
%\begin{scriptsize}
\xi_1(k+1)=&A_{\rm c}(k)\xi_1(k)+B_{\rm c}(k)u_0(k)\\[1mm]
u_r(k)=&C_{\rm c}(k)\xi_1(k)\\[3mm]
\xi_2(k+1)=&\Phi_2(k)\xi_2(k)+\Psi_2(k)(-u_r(k))\\[1mm]
u_{\rm im}(k)=&\Gamma_2(k)\xi_2(k)+D_2(k)(-u_r(k))\\[3mm]
\hat z(k+1)=&(F_{22}(k)-HF_{12}(k))\hat z(k)+(G_2-HG_1)u_{\rm st}(k)+\\[1mm]
&((F_{22}(k)-HF_{12}(k))H+F_{21}(k)-HF_{11}(k))e(k)\\[1mm]
\hat x_{\rm b}(k)=&\hat z(k)+He(k)\\[1mm]
u_{\rm st}(k)=&K_1(k)e(k)+K_2(k)\hat x_{\rm b}(k)\\[3mm]
\hat x(k+1)=&A_{\rm c}\hat x(k)+B_{\rm c}u(k)+E_{\rm c}\hat d_l(k)-L_1(\hat y(k)-y(k))\\[1mm]
\hat d_l(k+1)=&\hat d_l(k)-L_2(\hat y(k)-y(k))\\[1mm]
\hat y(k)=&C_{\rm c}\hat x(k)\\[3mm]
u_0(k)=&u_{\rm im}(k)+u_{\rm st}(k)\\[1mm]
u(k)=&u_0(k)-\dfrac{\hat d_l(k)}b\,,
%\end{scriptsize}
\end{align*}
where $(\Phi_{2}(\cdot),~\Psi_{2}(\cdot),~\Gamma_{2}(\cdot),~D_{2}(\cdot))$ is obtained by solving the Sylvester equation~(\ref{se}); $F$ and $G$ are defined in~(\ref{augsys}); $H$, $L_1$ and $L_2$ are observer gains;  and $K(k)$ is obtained by solving the LMIs~(\ref{lmi}).
\end{theorem}
\begin{proof}
The proof immediately follows by noting Lemma~\ref{l2},~\ref{stabilization} and~\ref{polytype}, Proposition~\ref{immersion}, the nominal system~(\ref{dnplant}), and the reduced order observer~(\ref{ob}).

\hfill $\square$
\end{proof}
\section{Simulation and experimental study}\label{exp}
To validate the tracking performance of the proposed robust design of the TV-IMPC method, simulation and experimental results are provided in this section. Specifically, the reference signal is generated by a time-varying exosystem as follows
\begin{equation*}\label{gd}
\begin{array}{rcl}
w(k+1)&=&\begin{bmatrix}1 &T_s(1+0.5\sin(2\pi t))\\T_s(-1+0.5\sin(5t)) & 1\end{bmatrix}w(k)\\[6mm]
r(k)&=&\begin{bmatrix}\lambda & 0\end{bmatrix}w(k)\,,
\end{array}
\end{equation*}
where the sampling interval $T_s=1{\rm ms}$, a constant gain $\lambda$, and two time-varying coefficients in the exosystem are with an irrational ratio ($\omega_1=2\pi$ and $\omega_2=5$).

The plot of the reference $r(k)$ ($\lambda=1$) is illustrated in Figure~\ref{re}. Indeed, it is not periodic and its amplitude varies in each cycle.
\begin{figure}[!ht]
  \centering
  \includegraphics[width=0.8\hsize]{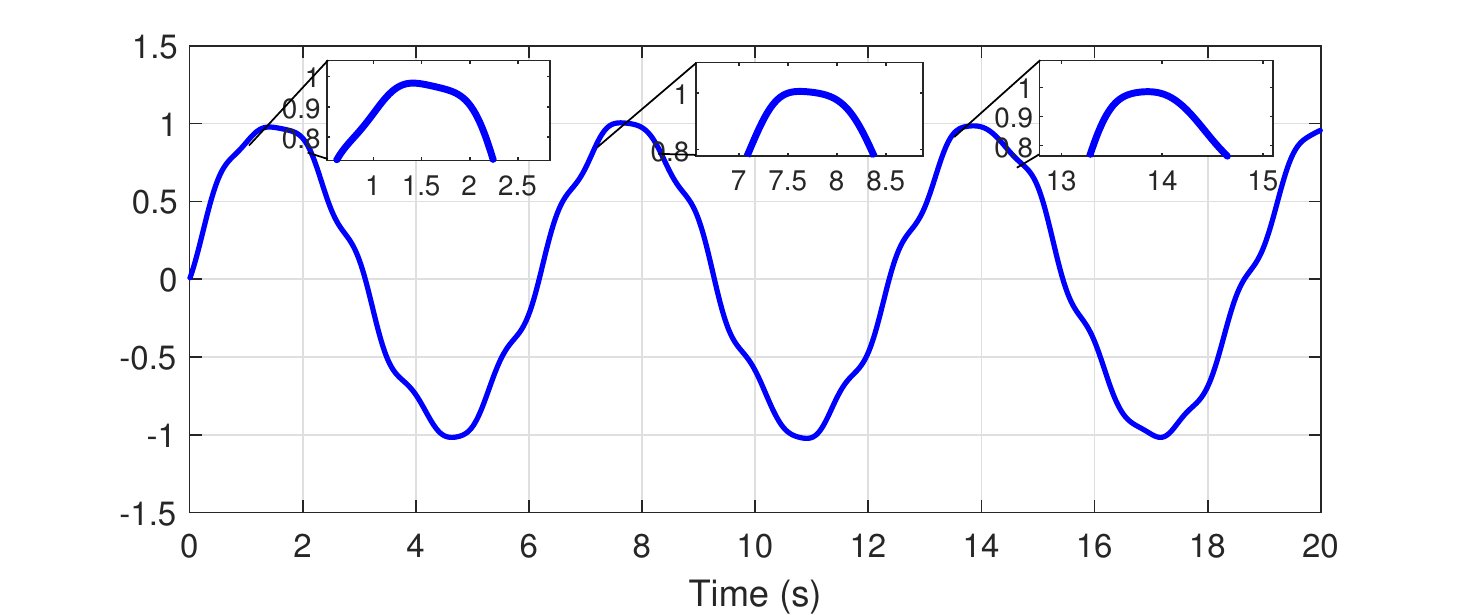}\\
  \caption{The plot of the reference $r(k)$.}\label{re}
\end{figure}
\subsection{Simulation study}
The simulations are conducted in MATLAB$^\text{TM}$ Simulink. The parameters of the plant model~(\ref{dplant}) are defined as follows
\begin{equation}\label{expxd}
\begin{split}
A_{\rm c}=&\begin{bmatrix}0 &1\\-0.9613 & 1.9404\end{bmatrix}\\[1mm]
B_{\rm c}=&\begin{bmatrix}0&1\end{bmatrix}^{\top}\\[1mm]
E_{\rm c}=&\begin{bmatrix}0&4.96{\rm e}^{-5}\end{bmatrix}^{\top}\\[1mm]
C_{\rm c}=&\begin{bmatrix}0.0098& 0.0099\end{bmatrix}\,.
\end{split}
\end{equation}
The dynamic parameters of the internal model controllers are calculated in real-time by equation~(\ref{se})-(\ref{cp}) with $K=\begin{bmatrix}-107.11&-69.37\end{bmatrix}$ solved by the LMI~(\ref{lmi}); the observer gains are $H=1{\rm e}^{-4}$, $L_1=\begin{bmatrix}96.71& 114.20\end{bmatrix}^{\top}$, and $L_2=2.75{\rm e}^{4}$.

To validate the  tracking performance of the robust TV-IMPC, the first simulation is conducted using a completely-known plant model, which indicates that $d_l(k)=0$. In the simulation results, it is seen from Figure~\ref{ideal} that the proposed method is able to achieve asymptotic tracking performance, and the Root-Mean-Square Error (RMSE) is of order $10^{-16}$ (at the precision of floating numbers in MATLAB).
\begin{figure}[!ht]
  \centering
  \includegraphics[width=0.8\hsize]{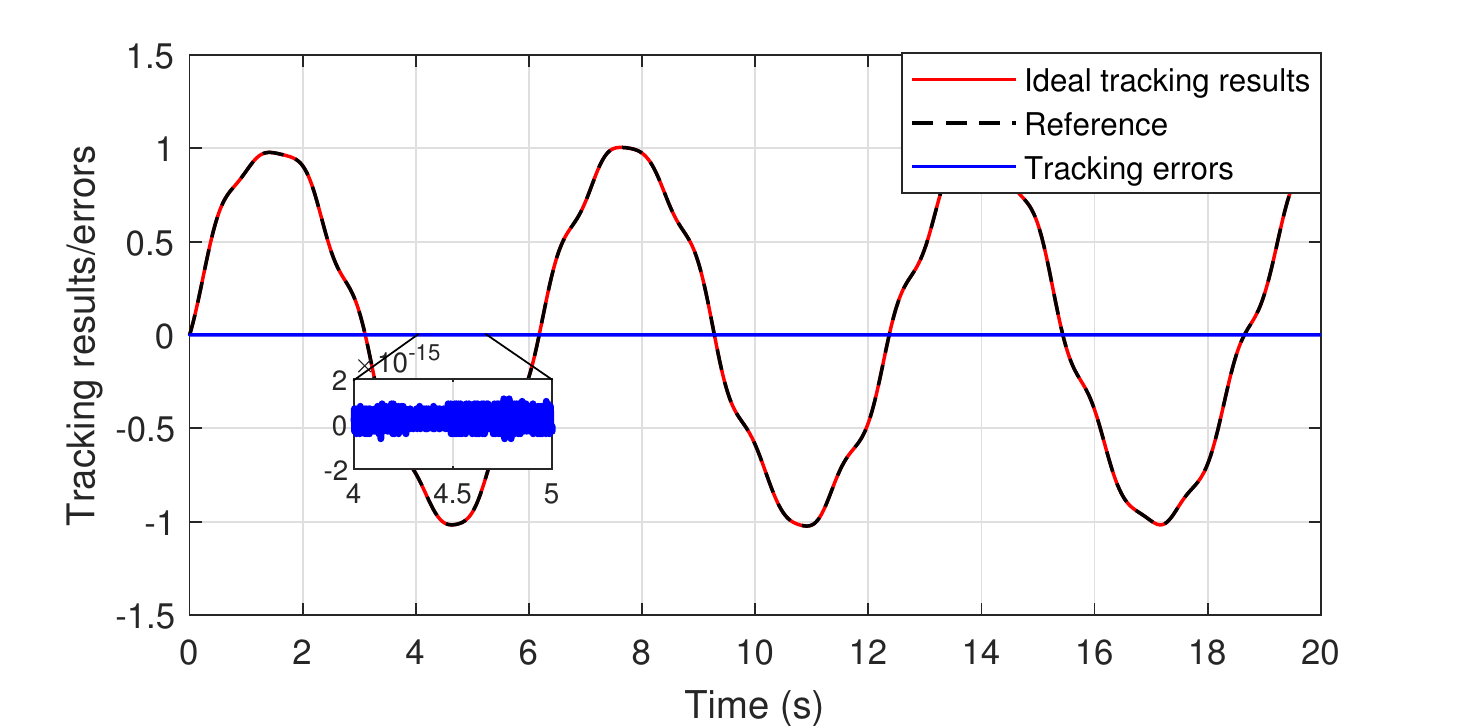}\\
  \caption{Simulation results of the proposed robust TV-IMPC without model uncertainties and external disturbances.}\label{ideal}
\end{figure}

To further validate the robustness of the proposed method, the signal $d_l(k)$ is introduced of the following form
\begin{equation}\label{d}
{d_l(k)=k_1\sin(k_2x^2_1(k)x_2(k))+n(k)\,,}
\end{equation}
where the first term of $d_l(k)$ represents the unmodeled dynamics with $k_1 = 1{\rm e}^3\sin(2\pi kT_s)$ and $k_2 = 1{\rm e}^{-4}$; and $n(k)$ is a square-wave-form noise signal with amplitude $1{\rm e}^{-2}$ and frequency $4\pi$. We investigate the effectiveness of our robust controller design by comparing the tracking errors with and without ESO. From the simulation results shown in Figure~\ref{fr}(a), it is seen that the tracking errors are significantly reduced, and the RMSE with ESO reaches $7.01{\rm e}^{-6}$, while the RMSE without ESO is $5.19{\rm e}^{-5}$. Moreover, Figure~\ref{fr}(b) shows the estimation results of the unknown term $d_l(k)$, validating that the ESO is able to achieve real-time estimation of $d_l(k)$, and the relative error of the estimation is $2.04\%$.
\begin{figure}[!ht]
	\centering
	\setcounter {subfigure} {0} \tiny (a){
		\includegraphics[scale=0.5]{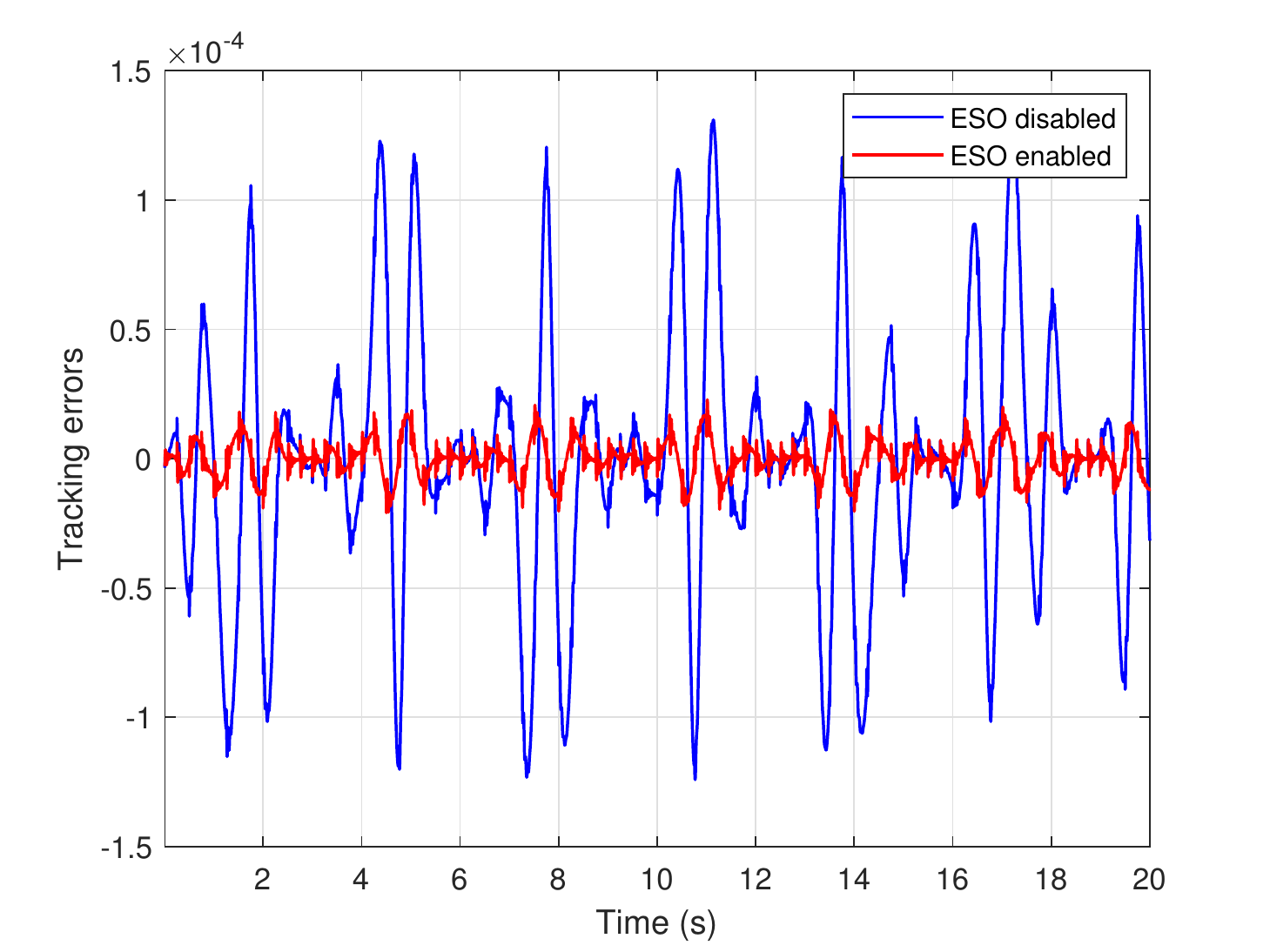}
        }
	\setcounter {subfigure} {0} \tiny (b){
		\includegraphics[scale=0.5]{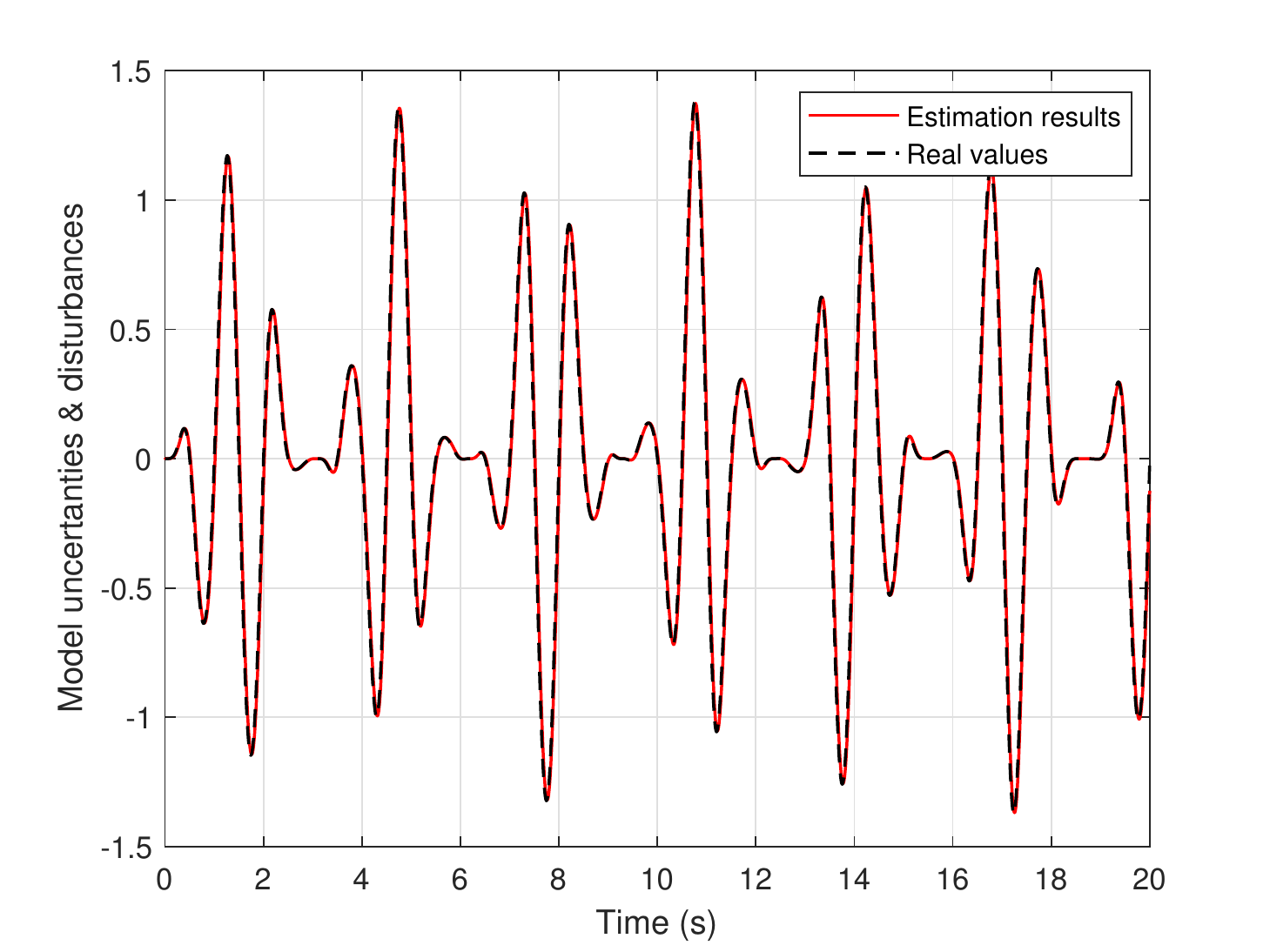}
        }
	\caption{Simulation results of the proposed robust TV-IMPC with model uncertainties and external disturbances. (a): Tracking errors (ESO enabled/disabled); (b): Estimation results of model uncertainties and disturbances.}
	\label{fr}
\end{figure}

Moreover, as the ESO gain $L_2$ is increased to $L_2=1.02{\rm e}^{6}$ (and $L_1$ is calculated as $\begin{bmatrix}100.52& 305.26\end{bmatrix}^{\top}$), it is found that the tracking errors are significantly decreased (RMS $1.62{\rm e}^{-6}$) than that of  the original case, as shown in Figure~\ref{ll2}. This  result validates that the tracking precision of the proposed robust TV-IMPC can be improved by adjusting the observer gains.
\begin{figure}[!ht]
  \centering
  \includegraphics[width=0.8\hsize]{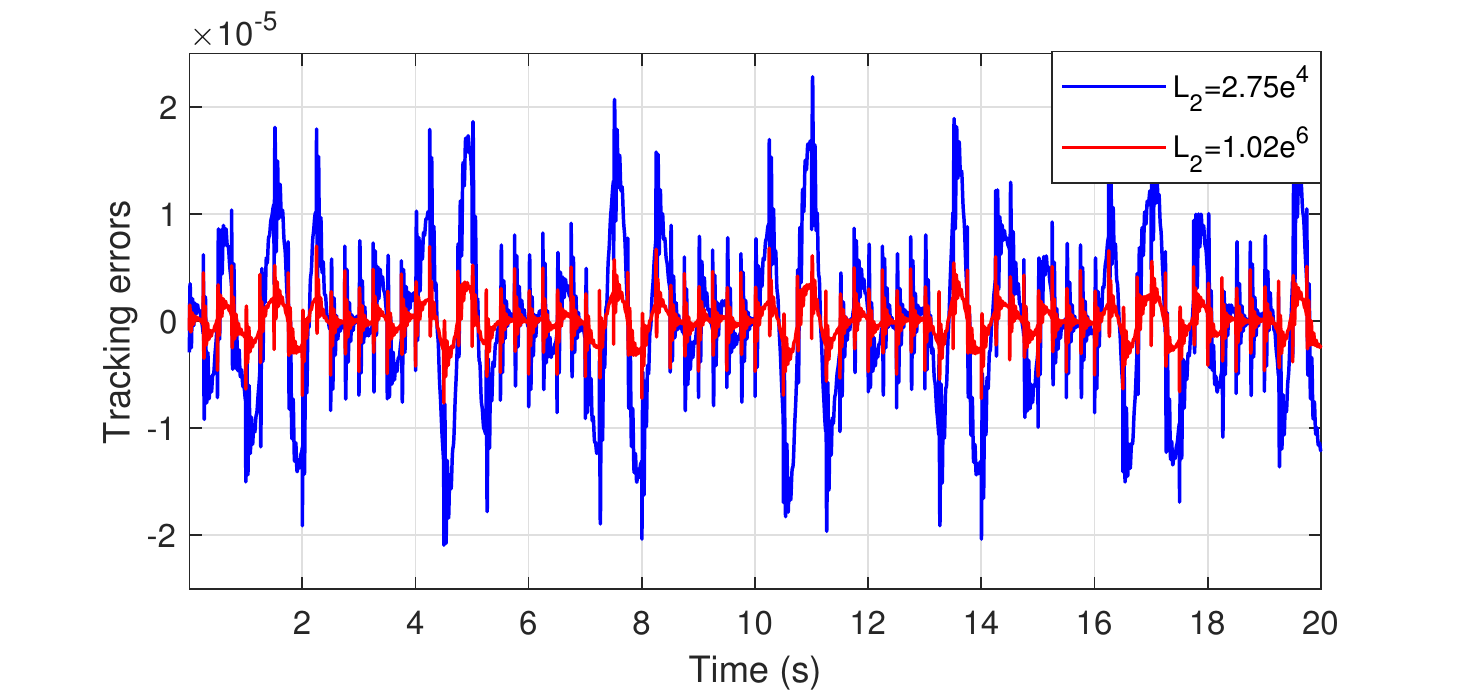}\\
  \caption{Simulation results of the tracking errors with different ESO gain $L_2$.}\label{ll2}
\end{figure}

Another comparative simulation is conducted by replacing the gray-box ESO with the black-box one, and the results are shown in Figure~\ref{bl}. It is noticed that the system becomes unstable with the same controller parameters, because the output saturation occurs. The above results verify that the information of plant dynamics should be considered in the ESO design to make the system easier to stabilize. Therefore, the gray-box ESO is more suitable for achieving the robustness of the TV-IMPC.
\begin{figure}[!ht]
  \centering
  \includegraphics[width=0.8\hsize]{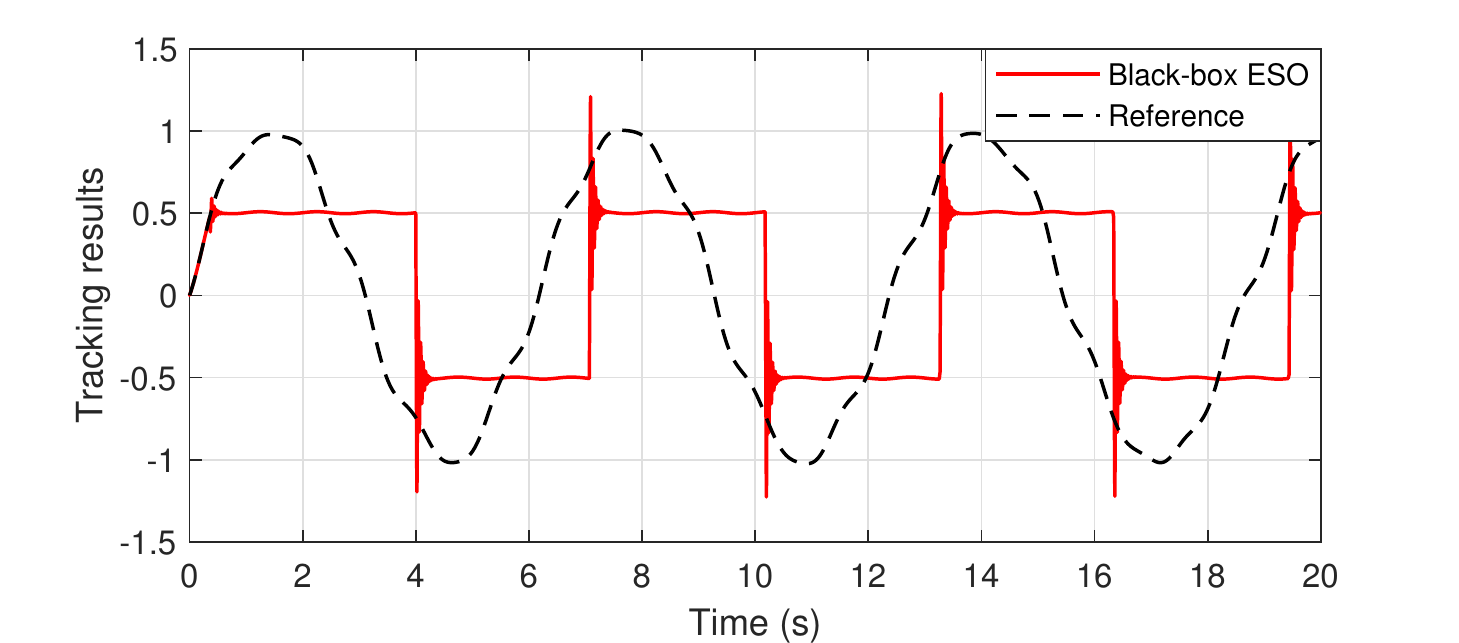}\\
  \caption{Simulation results of the proposed robust TV-IMPC with black-box ESO.}\label{bl}
\end{figure}
\subsection{Experimental study}
\begin{figure}[!ht]
  \centering
  \includegraphics[width=\hsize]{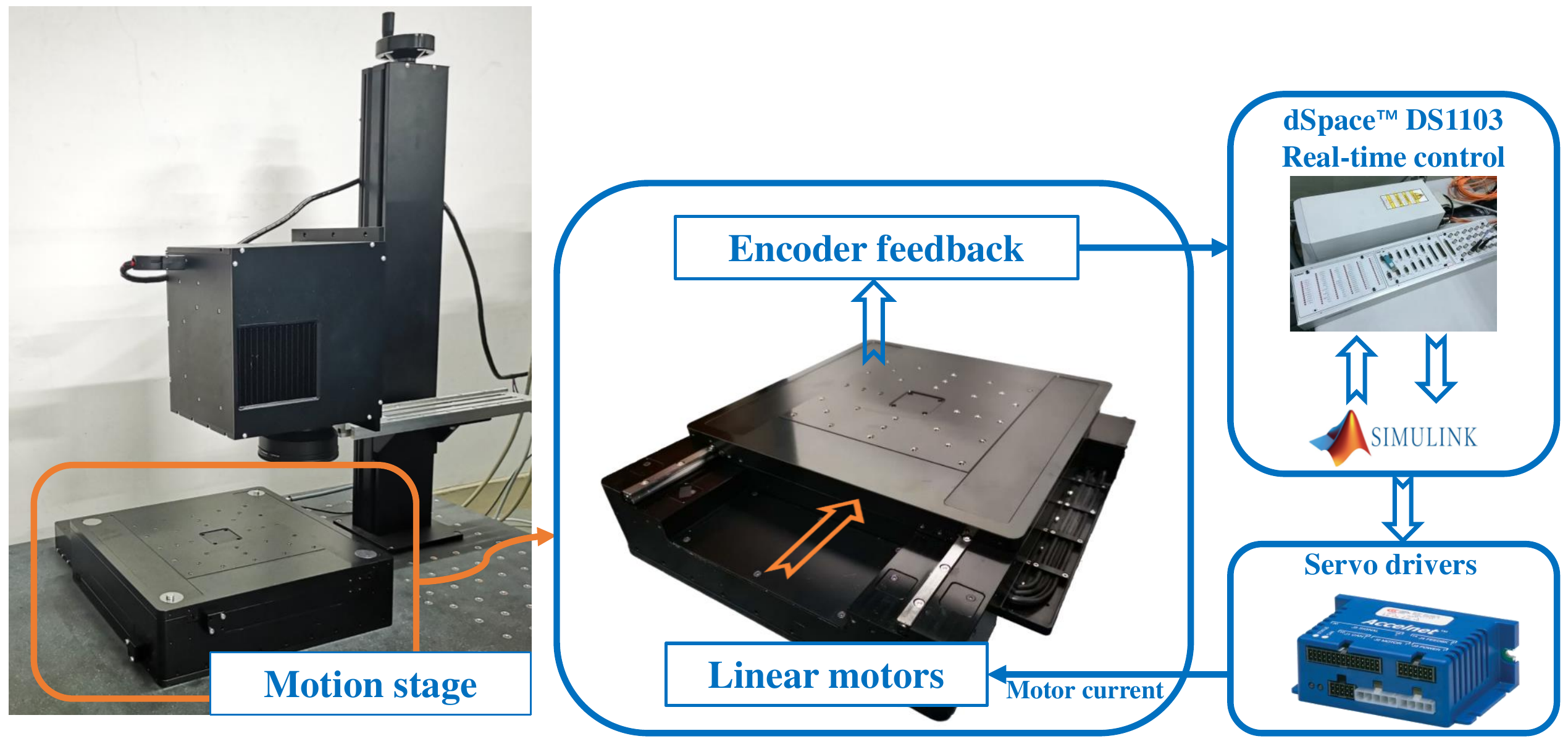}\\
  \caption{The experimental setup.}\label{stage}
\end{figure}
To demonstrate the feasibility of the proposed robust TV-IMPC design  in practice, experiments are implemented on the upper axis of a self-developed bi-axial direct-drive servo stage shown in Figure~\ref{stage}. The detailed parameters of the system are listed in Table.~\ref{stagep}. A dSPACE$^\text{TM}$ 1103 rapid prototyping control system  are used for controller implementation and real-time control executions.
\begin{table}[!ht]
  \centering
     \makeatletter\def\@captype{table}\makeatother\caption{Technical parameters of the servo stage.}
  \begin{tabular}{lccc}
    \toprule
    \tabincell{c}{\textbf{Actuators}} & Akribis AUM2-S-S4 \\
    \tabincell{c}{\textbf{Servo drivers}} & Copley Accelnet ACJ-090-12 \\
    \tabincell{c}{\textbf{Cross roller guides}} & NB SV4360-35Z \\
    \tabincell{c}{\textbf{Max. stroke}} & $\pm100$mm \\
    \tabincell{c}{\textbf{Max. velocity}}   & $300$mm/sec \\
    \tabincell{c}{\textbf{Max. acceleration}}   & $0.5$g\\
    \tabincell{c}{\textbf{Optical encoders}}   & Renishaw T1011-15A\\
    \tabincell{c}{\textbf{Encoder resolution}}   & $10$nm\\
    \bottomrule
  \end{tabular}
   \label{stagep}
\end{table}
\subsubsection{Modeling of the servo stage}
\begin{figure}[!ht]
  \centering
  \includegraphics[width=\hsize]{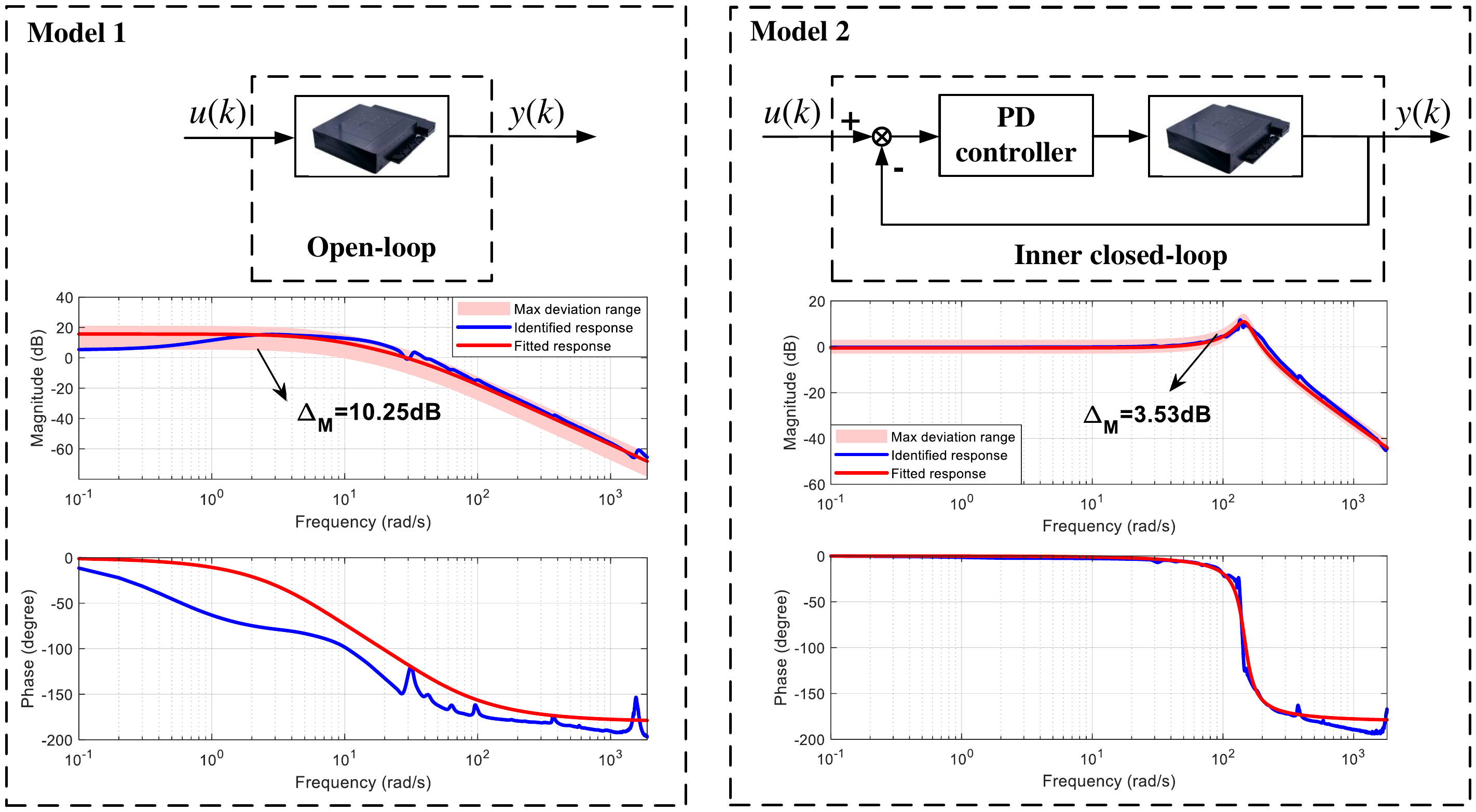}\\
  \caption{Different modeling methods and corresponding Bode plots of the identified/fitted model dynamics.}\label{bode}
\end{figure}
The nominal linear dynamics of the motion stage is modeled in two ways. The first one utilizes a direct open-loop approach (Model~1 in Figure~\ref{bode}), while the second one utilizes an indirect inner-closed-loop approach (Model~2 in the right of Figure~\ref{bode}).

In Model $1$, the maximum magnitude response deviation is $10.25$ dB. Also, there is a DC gain mismatch between the identified and the fitted response, which leads to infeasibility of the gray-box ESO design, because ESO requires the true value of the output gain $b$ in system~(\ref{dplant}).

In Model $2$, a PD controller ($K_{\rm p}=10,~\mbox{and}~K_{\rm d}=5{\rm e}^{-3}$) is embedded to form an inner closed-loop. As a result, the maximum magnitude response deviation is relatively small ($3.53$ dB). Moreover, the fitted response curve at low frequency is close to the identified one. Therefore, the gray-box ESO can be used for Model $2$ to further compensate for the high-frequency model uncertainties and to ensure the structural robustness of the TV-IMPC.

The detailed dynamics of Model $2$ is listed in~(\ref{expxd}) of the simulation section. Thus, the experimental parameter design is the same as that used in the simulations.

The reference signal of the experiments is magnified ($\lambda=80$) to better demonstrate the tracking precision of the proposed robust TV-IMPC, i.e.,
\begin{equation*}\label{gd2}
\begin{array}{rcl}
w(k+1)&=&\begin{bmatrix}1 &T_s(1+0.5\sin(2\pi t))\\T_s(-1+0.5\sin(5t)) & 1\end{bmatrix}w(k)\\[6mm]
r(k)&=&\begin{bmatrix}80 & 0\end{bmatrix}w(k)\,,
\end{array}
\end{equation*}
which requires a motion stroke of around $\pm 80$mm.

%\begin{figure}[ht]
%\centering
%\begin{minipage}[t]{0.48\textwidth}
%\centering
%\includegraphics[width=0.9\hsize]{expresult.eps}
%\end{minipage}
%\begin{minipage}[t]{0.48\textwidth}
%\centering
%\includegraphics[width=0.9\hsize]{experror.eps}
%\end{minipage}
%\caption{Experimental tracking results of the proposed TV-IMPC. Left: Tracking results; Right: Tracking errors (with the ESO enabled/disabled).}\label{expr}
%\end{figure}
%\begin{figure}[ht]
%  \centering
%  \includegraphics[width=0.48\hsize]{er.eps}\\
%  \caption{Experimental model uncertainties and disturbances estimation results of the proposed robust TV-IMPC.}\label{er}
%\end{figure}
\subsubsection{Experimental results}
The experimental studies are divided into three parts, which will be discussed in sequence as follows.

1) Comparison of tracking results of the TV-IMPC with/without ESO.
\begin{figure}[!ht]
  \centering
  \includegraphics[width=0.8\hsize]{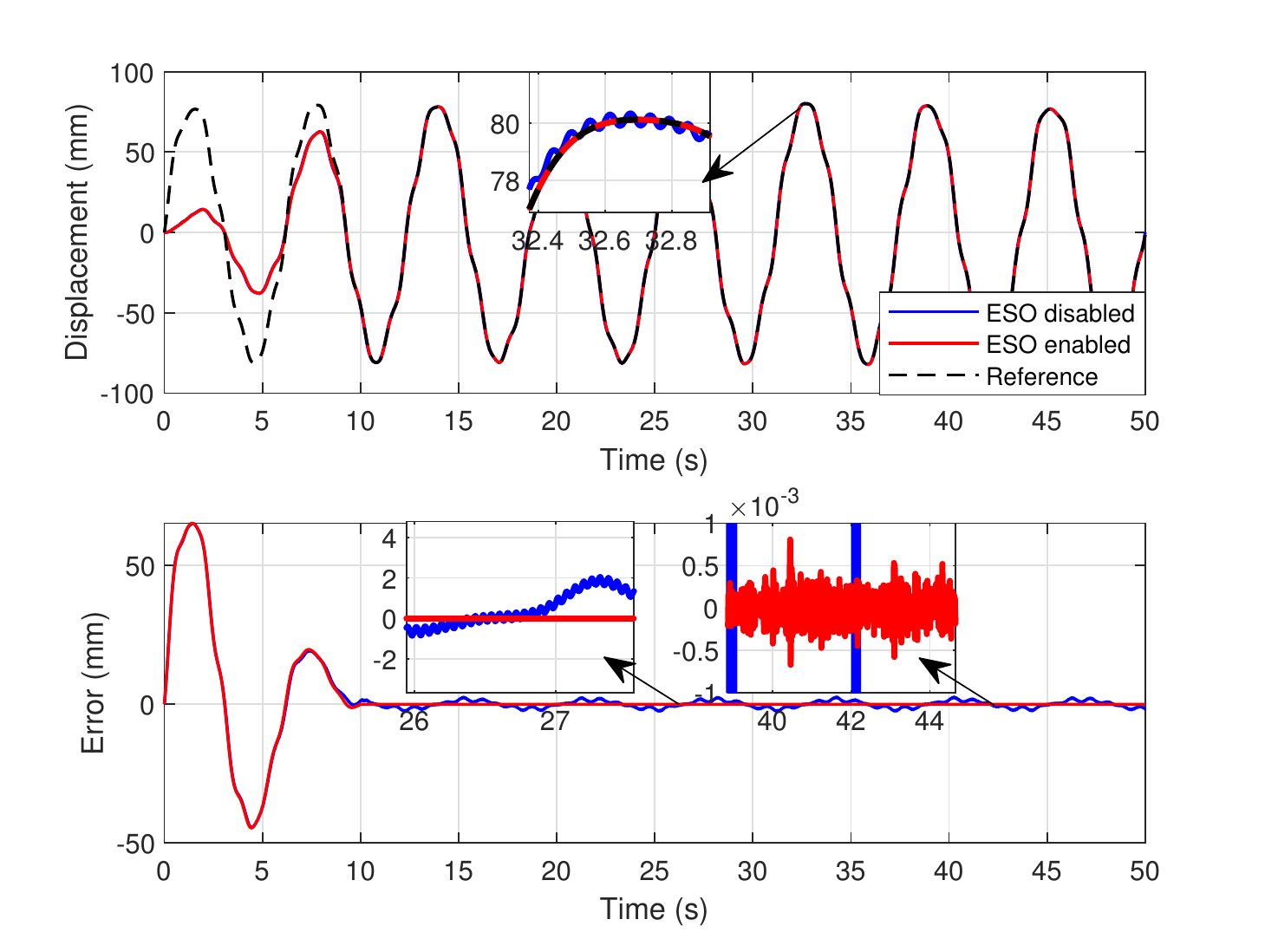}\\
  \caption{Tracking results of $r(k)$ by the TV-IMPC with/without ESO.}\label{noeso}
\end{figure}

The tracking experiments of the plant model 2 and the reference $r(k)$ are implemented respectively with the ESO feedback switched on/off. The tracking results and errors are plotted in Figure~\ref{noeso}.

Without the ESO, it is observed that the tracking errors of Model $2$ exhibit significant chattering. Indeed, the overall TV-IMPC control scheme for plant model $2$ is dual-loop. The high-frequency deviation between the nominal model and the real model violates Proposition~\ref{io}, thus the tracking cannot be asymptotic. The chattering arises due to the transient performance mismatch between the inner-loop bandwidth of the PD controller and the outer-loop bandwidth of the TV-IMPC controller.

In contrast, the ESO is capable of estimating and compensating for high-frequency model uncertainties and potential disturbances simultaneously. Therefore, the tracking errors with the ESO are much smaller and more stationary. Meanwhile, thanks to the time-varying internal model, the structure of the reference is not observed from the errors, and the relative error can be reduced to the level of $10^{-6}$.

2) Comparison of tracking results of the robust TV-IMPC with different ESO gain $L_2$.

We conduct experiments using two different values for the ESO gain $L_2$, namely $2.75{\rm e}^4$ and $1.02{\rm e}^6$, similar to the simulation parameters. The tracking errors for each value are shown in Figure~\ref{l2d}. It is found that increasing the ESO gain resulted in even better tracking performance, which is consistent with both theoretical and simulation results. With the larger $L_2$ value, the steady-state tracking RMSE is reduced to $61.14$nm, with a relative error of $1.08{\rm e}^{-6}$ and a maximum tracking error of $283.28$nm.
\begin{figure}[!ht]
  \centering
  \includegraphics[width=0.8\hsize]{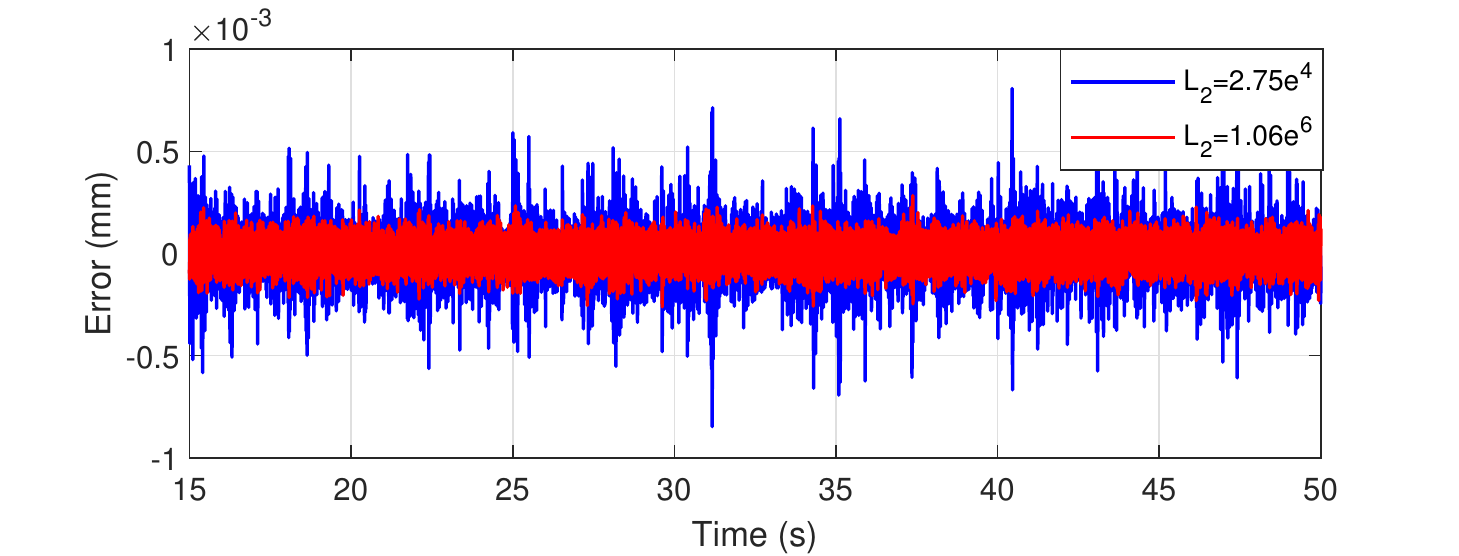}\\
  \caption{Tracking errors of $r(k)$ by the robust TV-IMPC with different ESO gain $L_2$.}\label{l2d}
\end{figure}

3) Comparison of tracking results of different plant models: Model $1$ and Model $2$.

Although the ESO structure cannot be used for Model $1$ due to the mismatch in the input gain $b$, we can compare the tracking performance of two different cases in this experiment,
\begin{itemize}
\item Case 1: Based on Model $1$ and the TV-IMPC method (without ESO);
\item Case 2: Based on Model $2$ and the robust TV-IMPC method (with ESO and $L_2=1.02{\rm e}^6$).
\end{itemize}

The tracking results of Cases $1$ and $2$ are shown in Figure~\ref{open}. It is observed that Case $1$ has a faster transient response compared to Case $2$, because the dual-loop control in Case $2$ narrows its bandwidth. However, the tracking errors of Case $1$ cannot converge to the level of stationary noise due to the unmodeled dynamics shown in the Bode plot of Figure~\ref{bode}.
\begin{figure}[!ht]
  \centering
  \includegraphics[width=0.8\hsize]{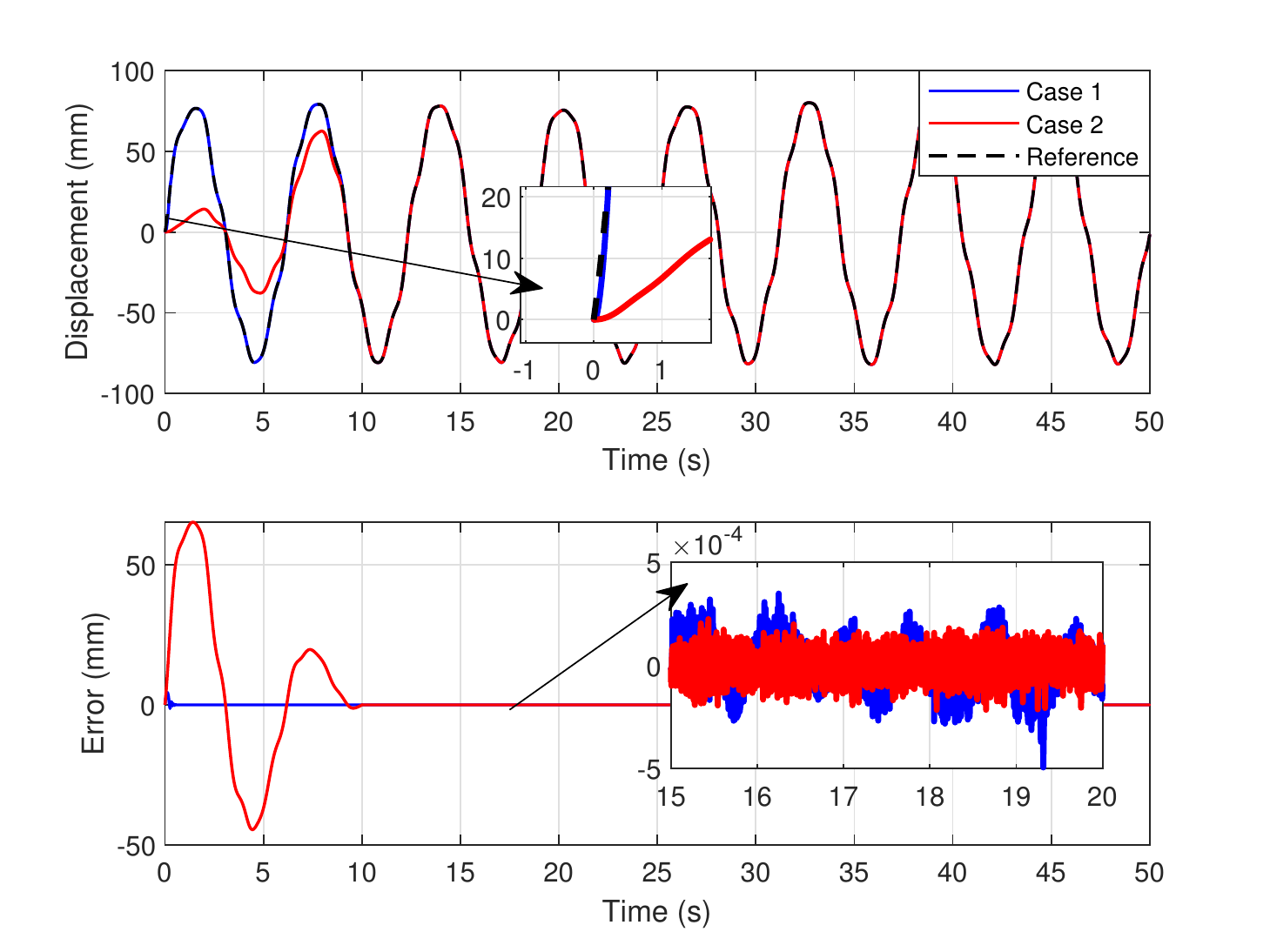}\\
  \caption{Tracking results of $r(k)$ by the TV-IMPC in Case $1$ and $2$.}\label{open}
\end{figure}

In summary, the robust TV-IMPC method outperforms the existing method in achieving ultra-precision tracking. The steady-state tracking RMSE ($\varepsilon_{\rm rms}$), maximum error ($\varepsilon_{\rm max}$), and relative error ($\delta_{\rm r}$) are listed in Table~\ref{indices}.
\begin{table}[!ht]
  \small
  \centering
     \makeatletter\def\@captype{table}\makeatother\caption{Experimental tracking performance indices}
  \begin{tabular}{c|c|ccc}
    \toprule
    Plant model & Model $1$ & \multicolumn{3}{c}{Model 2}\\[2mm]
    Control design & without ESO & without ESO & $L_2=2.75{\rm e}^4$ & $L_2=1.02{\rm e}^6$\\
    \midrule
    $\varepsilon_{\rm rms}$ [nm]& $112.89$ & $1.22{\rm e}^6$ & $135.77$ & $61.14$ \\[2mm]
    $\varepsilon_{\rm max}$ [nm]& $516.65$ & $2.66{\rm e}^6$ & $847.79$ & $283.28$ \\[2mm]
    $\delta_{\rm r}$ & $1.99{\rm e}^{-6}$ & $2.16{\rm e}^{-2}$ & $2.40{\rm e}^{-6}$ & $1.08{\rm e}^{-6}$ \\
    \bottomrule
  \end{tabular}
   \label{indices}
\end{table}
\section{Conclusion}\label{conclusion}
In this paper, we have proposed a robust design for the TV-IMPC method that utilizes the ESO to obtain the structural robustness of the time-varying internal model in the presence of plant model uncertainties and external disturbances.
%The design consists of a time-varying internal model controller and a compensator enabled by the gray-box ESO.
In particular, the ESO is designed in the gray-box fashion to better estimate the potential model uncertainties and external disturbances. Furthermore, the boundedness of the ESO estimation error and the additive system uncertainty are analyzed, and hence a time-varying stabilizer can be developed for the augmented system of the time-varying internal model and ESO compensator.
Extensive simulation and experimental studies are conduced on a direct-drive servo stage and the results validate that the proposed robust TV-IMPC significantly improves the tracking precision compared to the existing TV-IMPC.
\section*{Acknowledgments}\label{ack}
This work was supported by the National Natural Science Foundation of China [Grant numbers 52275564, 51875313].

\bibliographystyle{unsrt}
%\bibliography{references}  %%% Remove comment to use the external .bib file (using bibtex).
%%% and comment out the ``thebibliography'' section.

\bibliography{references}

\end{document}